\documentclass[10pt]{llncs}

\usepackage{color}
\usepackage{url}

\usepackage[usenames,dvipsnames]{xcolor}
\usepackage{multirow}
\usepackage{graphics}
\usepackage{a4wide}

\usepackage{epstopdf}

\newcommand{\cF}{\mathcal{F}}

\newcommand{\cs}[1]{\textcolor{black}{#1}}

\newcommand{\steven}[1]{\textcolor{black}{#1}}
\newcommand{\smore}[1]{\textcolor{black}{#1}}
\newcommand{\sblah}[1]{\textcolor{black}{#1}}
\newcommand{\commentOut}[1]{}

\newcommand{\sthursday}[1]{\textcolor{black}{#1}}
\newcommand{\cswg}[1]{\textcolor{black}{#1}}

\usepackage{amsmath}
\usepackage{graphicx}

\title{Phylogenetic incongruence through the lens of Monadic Second Order logic}

\author{Steven Kelk, Leo van Iersel and Celine Scornavacca}

\institute{}
\institute{
Department of Knowledge Engineering (DKE), Maastricht University, P.O. Box 616, 6200 MD Maastricht, The Netherlands, \email{steven.kelk@maastrichtuniversity.nl}
\and
\smore{Delft Institute of Applied Mathematics, Delft University of Technology, P.O. Box 5, 2600 AA Delft, The Netherlands, \email{l.j.j.v.iersel@gmail.com}}
\and
Institut des Sciences de l'Evolution (Universit\'e de Montpellier, CNRS, IRD, EPHE), Place E. Bataillon - CC 064 - 34095 Montpellier Cedex 5, France, \email{celine.scornavacca@univ-montp2.fr}
}

\begin{document}
\maketitle

\begin{abstract}
Within the field of phylogenetics there is growing interest in measures for summarising the dissimilarity, or \emph{incongruence}, of two or more phylogenetic trees.  Many of these
measures are NP-hard to compute and this has stimulated a considerable volume of research into fixed parameter tractable algorithms.
In this article we use \emph{Monadic Second Order} logic \smore{(MSOL)} to give alternative, compact proofs of fixed parameter tractability for several well-known incongruency measures.
In doing so we wish to demonstrate the considerable potential of MSOL - machinery still largely unknown outside the algorithmic graph theory community - within phylogenetics.
A crucial component of this work is the observation that many of these measures, when bounded, imply the existence of an \emph{agreement forest} of bounded size, which in turn implies that an auxiliary graph structure, the \emph{display graph}, has bounded treewidth. It is this bound on treewidth that makes the machinery of MSOL available for proving fixed parameter tractability. We give a variety of different MSOL formulations. Some
are based on explicitly encoding agreement forests, while some only use \cswg{them} implicitly to generate the treewidth bound. Our formulations introduce a number of \smore{``phylogenetics MSOL primitives'}' which will hopefully be of use to other researchers.


\end{abstract}

\section{Introduction}

The central goal of phylogenetics is to accurately infer the evolutionary history of a set of species (or \emph{taxa}) $X$ from incomplete information. Classically, phylogenetic reconstruction has access to  information about each element in $X$, such as DNA data, and seeks to infer a phylogenetic tree - a tree whose leaves are bijectively labeled by $X$ - that best fits this data. There is a vast literature available on this topic and many different algorithms \smore{exist} for constructing phylogenetic trees \cite{felsenstein2004inferring,SempleSteel2003}. In practice, it is not uncommon for phylogenetic analysis to generate multiple phylogenetic trees as output.
This can occur for various reasons, ranging from software engineering choices (many tree-building packages are designed to generate multiple optimal and near-optimal solutions) to more structural explanations (\emph{reticulate} evolutionary signals that are comprised of multiple distinct tree signals). Given two (or more) distinct phylogenetic trees, it is natural to compare them to determine whether the difference is significant. This \cswg{explains the interest of} the
phylogenetics community for measures that can quantify the dissimilarity, or \emph{incongruence}, of phylogenetic trees \cite{Huson2010}. Some of these measures (such as \smore{\textsc{Tree Bisection and \sthursday{Reconnection}}} distance \cite{AllenSteel2001}) are studied to better understand how local-search heuristics, based on rearrangement operations, navigate
the space of phylogenetic trees (e.g., \cite{bryant2004splits}). Others, such as \textsc{Hybridization Number} \cite{Bordewich2007}, are studied because they assist with the inference of phylogenetic networks, which generalise
phylogenetic trees to directed acyclic graphs \cite{Huson2010,Huson2011}.

Unfortunately, many of these measures are NP-hard and APX-hard to compute. On the
positive side, however, the phylogenetics community has been quite successful in
proving that these measures are fixed parameter tractable (FPT) in their natural parameterizations. Informally this means
that a measure \smore{that evaluates to} $k$ can be computed in time $f(k) \cdot poly(n)$ where
$f$ is some function that only depends on $k$ and $n$ is the size of the instance (often
taken to be $|X|$). Such running times have the potential to be much faster than running
times of the form $O( n^{f(k)} )$ when the measure in question is comparatively small.
See e.g. \cite{downey2013fundamentals} for more background on FPT.  A number
of state-of-the-art phylogenetics software packages are based on FPT algorithms, \smore{such as the software used in \cite{Whidden2014}}. Most
FPT results in the phylogenetics literature are based on classical proof techniques such
as polynomial-time kernelization and bounded-search. 

Parallel to all of this, algorithmic graph theorists have made great steps forward in
identifying sufficient, structural conditions under which NP-hard problems on graphs become
(fixed parameter) tractable. At the heart of this research lies the width parameter, the
most famous example being \emph{treewidth}. Informally treewidth is a measure that
quantifies the dissimilarity of a graph from being a tree. The
notion of treewidth, which is most famously associated with the celebrated Graph Minors
project of Robertson and Seymour \cite{robertson1986graph}, has had a profound impact upon algorithm design. A great many NP-hard problems turn out to become tractable on graphs of bounded treewidth, using broadly similar proof techniques i.e. dynamic programming on tree decompositions \cite{bodlaender1994tourist}. This contributed to the rise of meta-theorem\cs{s}, the archetypal example being \emph{Courcelle's Theorem} \cite{Courcelle90,Arnborg91}. 
This states, \smore{when combined with the result from \cite{Bodlaender96}}, that any graph property that can be abstractly formulated as a length $\ell$
sentence of \emph{Monadic Second Order} logic (MSOL), can be tested in time $f(t,\ell) \cdot O(n)$ on graphs of treewidth $t$, where $n$ is the number of vertices in the graph. When
$t$ and $\ell$ are both bounded by a function of a single parameter $p$, this yields
a running time of the form $f(p) \cdot O(n)$ i.e. linear-time fixed parameter tractability in parameter $p$. This is an extremely powerful technique in the sense that it completely
abstracts away from ad-hoc algorithm design and permits \smore{highly} compact, ``declarative'' proofs that a problem is FPT.
Courcelle's Theorem (and its variants) are more than two decades old, but  \cs{their potential is rarely exploited by}
 the phylogenetics community. One exception is the literature on \emph{unrooted compatibility}, which asks whether a set of unrooted phylogenetic trees are compatible. The FPT proof by
Bryant and Lagergren \cite{compatibility2006} proves that the \emph{display graph} (the graph obtained by identifying all taxa with the same label) has bounded treewidth (in the number of input trees), and then gives an MSOL formulation which tests compatibility. A follow-up result by the present authors applies
a similar approach \cite{strictcompatibility2014}.

In this article we show that this technique has much broader potential within phylogenetics. To clarify the exposition we focus on binary trees (both rooted and unrooted) on the same set of taxa $X$. We begin by proving that if two trees have an \emph{agreement forest} of size $k$ --
essentially a partition of the trees into $k$ isomorphic subtrees -- the treewidth of the display graph is bounded by a function of $k$. This simple observation is significant because of the prominent role of agreement forests within the phylogenetics literature. We use this insight
to re-analyse three well-known NP-hard phylogenetics problems that were previously shown to be FPT using more conventional analysis. In particular, we give MSOL formulations for (1)
\textsc{Unrooted Maximum Agreement Forest} (uMAF), which is equivalent to the problem of computing \textsc{Tree Bisection and Reconnection} distance (TBR) on unrooted trees, (2)
\textsc{rooted Maximum Agreement Forest} (rMAF), which is equivalent to the problem
of computing \textsc{Rooted Subtree Prune and Regraft} distance (rSPR) on rooted trees, and (3) \textsc{Hybridization Number} (HN) on rooted trees. The formulations for
uMAF and rMAF are based on explicitly modelling agreement forests using quartets and edge cuts. The formulation for HN uses agreement forests implicitly to obtain the treewidth bound but, due to the
difficulties in encoding \emph{acyclic} agreement forests, then bypasses the agreement forest abstraction. Instead, it encodes an equivalent, ``elimination ordering'' formulation of HN which considers sequences of pruned common subtrees. Finally we consider the (4) \textsc{Maximum Parsimony Distance on Binary Characters} \cs{problem}.
This asks for a binary character $f$ on $X$ that maximizes the absolute difference between the parsimony score of $f$ on the two trees.
It is NP-hard but not known to be FPT (in \cs{the} parsimony distance). Here we give an optimization MSOL formulation which shows that the problem is FPT in parameter uMAF. Although this does not settle whether the natural parameterization of the problem is FPT, it does demonstrate a number of interesting principles. Firstly, it
demonstrates the power of ``simulating'' the execution of polynomial-time algorithms (in this case, Fitch's algorithm \cite{fitch1971}) within MSOL.
Secondly, any subsequent proof that TBR distance
is at most a bounded distance above $d^{2}_{MP}$ distance and/or that $d^2_{MP}$ distance induces bounded
treewidth display graphs, will automatically prove that $d^2_{MP}$ distance is FPT in its natural parameterization.

Summarizing, our formulations show the potential for MSOL to generate compact, logical
FPT proofs for phylogenetics problems. The machinery of MSOL does not yield practical algorithms but it \sblah{is an excellent classification tool. Once the existence of FPT algorithms has been confirmed via MSOL one can then switch efforts
to finding a \emph{good} FPT algorithm by more direct analysis, possibly (but not exclusively) through direct analysis of tree decompositions.}
Our formulations also introduce a number of phylogenetics ``primitives'' concerning quartets, clusters, subtrees and compatibility that we hope will be of use to other
phylogenetics researchers.


\section{Preliminaries}
In this section, we define the main objects that will be manipulated in this paper.  

An \emph{unrooted phylogenetic tree} $T$ (\cs{unrooted  tree for short}) \steven{is a tree in which no vertex has degree 2
and in which the leaves are  bijectively \sthursday{labeled} by a label set $\mathcal{L}(T)$. The
leaf labels are often called \emph{taxa} and the symbol $X$ is frequently used as shorthand
for $\mathcal{L}(T)$.  Internal vertices are not labeled. A \emph{rooted} phylogenetic tree (\cs{rooted  tree for short}) is defined similarly, except that it has exactly one vertex, called the \emph{root} of the tree, that \smore{is permitted to have} degree 2, and edges are directed away from the root. An unrooted tree is \emph{binary} if every internal
vertex has degree 3, and a rooted tree is binary if each internal vertex has indegree 1
and outdegree 2, and the root has outdegree 2 and indegree 0.}


Given  an unrooted  tree $T$ and a subset $Y \subseteq \mathcal{L}(T)$, 
\cs{we use $T(Y)$ to denote the minimal subtree of $T$ connecting $Y$. Moreover, we denote by  $T |_Y$  the tree obtained from $T(Y)$} when suppressing vertices of degree 2. We say that $T |_Y$ is the subtree of $T$ \emph{induced} by $Y$. \sblah{In graph theory terms, $T |_Y$ is a label-preserving topological minor of $T$.} Induced subtrees are defined in the same way for rooted trees, except that the root of $T |_Y$ becomes the vertex in the minimal connecting subgraph that is closest to the root of $T$, and we suppress all degree 2 vertices except the new root. \steven{We write $T - Y$ to denote
$T|_{\mathcal{L(T)} - Y}$.} For any node $u$ of a rooted tree $T$, $T_u$ is the subtree of $T$ rooted at $u$. 

\begin{figure}
\centering
\includegraphics[scale=0.2]{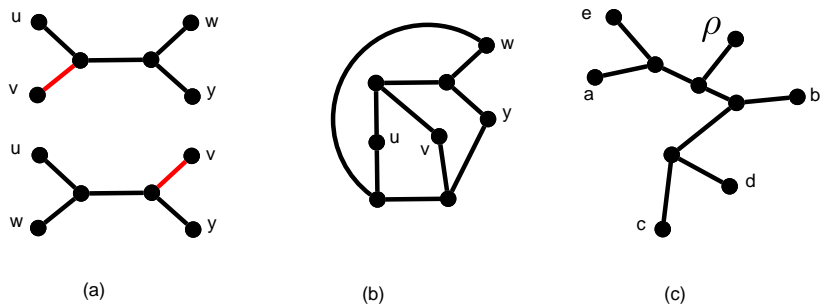}
\caption{\emph{(a) Two unrooted binary phylogenetic trees on $\{u,v,w,y\}$.  A maximum agreement forest (uMAF) for these two trees contains 2 components, and can be \sthursday{obtained} by cutting the single red edge in both trees and then suppressing the resulting degree 2 vertices. (b) The display graph for the two trees from (a), obtained
by identifying leaves with the same label. (c) How addition of a special label $\rho$ can be used to root a tree: edges are assumed to be directed away from $\rho$.}}
\label{fig:displaygraph} 
\end{figure}

Given a label set $X$, a \emph{bipartition} (or \emph{split}) $A|B$ on  $X$  is a partition of  $X$ into two non-empty sets. Each edge $\{u,v\}$ of a tree $T$ induces a split \steven{$\mathcal{L}(T_u)|\mathcal{L}(T_v)$}, where $T_u$ and $T_v$ are the two
trees obtained from $T$ when $\{u,v\}$ is deleted. \smore{Given a rooted tree $T$ \smore{with label set $X$}, a subset $X'$ of $X$ is called a \emph{clade} (or \emph{cluster}) of $T$, if $T$ contains a node $v$ such that $\mathcal{L}(T_v) = X'$.}


Given \steven{an unrooted binary} tree $T$ and a set of four \steven{distinct} labels $\{u,v,w,y\}$ in  $\mathcal{L}(T)$, $T|_{\{u,v,w,y\}}$ \steven{will be exactly one of} the three possible unrooted binary trees on $\{u,v,w,y\}$. 
These are called \emph{quartets} and are denoted respectively by $uv|wy$, $uw|vy$ and $wv|uy$, depending on
\steven{the bipartition induced by its central edge}. \sthursday{In Figure \ref{fig:displaygraph}(a) we see
$uv|wy$ and $uw|vy$.}
Given a rooted \steven{binary} tree $T$ and a set of three labels $\{u,v,w\}$ in  $\mathcal{L}(T)$, $T|_{\{u,v,w\}}$ \steven{will be exactly one of the three} possible rooted binary trees on $\{u,v,w\}$. 
These are called \steven{\emph{triplets}} and are denoted respectively by $uv|w$, $uw|v$ and $wv|u$, \steven{where $ij|k$ means that
the leaf labelled $k$ is incident to the root}.


Let $\mathcal{T}=\{T_1, T_2, \ldots, T_k\}$ be a collection of unrooted  trees, not necessarily on the same \steven{set of taxa}. The \emph{display graph} \smore{of} $\mathcal{T}$ is obtained from the disjoint graph union of all trees in $\mathcal{T}$ by identifying vertices with the same label; \sthursday{see Figure \ref{fig:displaygraph}(b)}.


Given an undirected graph $G = (V,E)$, a \emph{bag} is simply a subset of $V$. A \emph{tree decomposition} of $G$ consists of a tree $T_{G} = (V(T_G), E(T_G))$ where $V(T_G)$
is a collection of bags such that
the following holds: (1) every vertex of
$V$ is in at least one bag; (2) for each edge $\{u,v\} \in E$, there exists some bag that
contains both $u$ and $v$; (3) for each vertex $u \in V$, the bags that contain $u$ induce
a connected subtree of $T_G$. The \emph{width} of a tree decomposition is equal to the
cardinality of its largest bag, minus 1. The \emph{treewidth} of a graph $G$ is equal
to the minimum width, ranging over all possible tree decompositions of $G$. A tree with at
least one edge has treewidth 1. For a fixed value of $k$ one can determine in linear time
whether a graph has treewidth at most $k$ \cite{Bodlaender96}.
\section{Main results}

\sthursday{Unless} stated otherwise, we assume that $T_1 = (V_1, E_1)$ and $T_2 = (V_2, E_2)$ are both
unrooted binary trees on $X$. \steven{Their display graph is denoted by $D=(V,E)$ and $R^D$
denotes the vertex-edge incidence relation in $D$. We use $adj$ to denote the
vertex-vertex adjacency relation in $D$. Note that $|V| = 3|X|-4$ and $|E|=4|X|-6$.\commentOut{, i.e.
the number of vertices and edges in $D$ is linearly bounded by a function of $|X|$.}}
\subsection{TBR / MAF on unrooted trees}
\label{subsec:tbr}

We will start by giving the definitions of  a TBR move and of the TBR distance between two unrooted binary trees.
\begin{definition}[TBR move]
Given an unrooted binary tree $T$,  a \emph{tree bisection and reconnection} (TBR) move on $T$ \steven{consists of removing} an edge  of $T$, say $\{u,v\}$, and then reconnecting  the subtrees $T_u$ and $T_v$ as follows: subdividing an edge of $T_u$ with a new vertex
$p$; subdividing an edge of $T_v$ with a new vertex $q$; connecting $p$ to $q$;
and finally suppressing any vertices of degree 2. 
\end{definition}

\steven{\emph{TBR distance} is then defined naturally as follows:}\\
\\
\textbf{Problem: $d_{TBR}(T_1, T_2)$}\\
\textbf{Input: }Two unrooted binary trees $T_1$, $T_2$ on the same set of taxa $X$.\\
\textbf{Output: }The minimum number of TBR \sthursday{moves} required to transform $T_1$ into $T_2$.\\

We will now give the definition of an uMAF for two unrooted binary trees $T_1$, $T_2$ on $X$. 
Any collection of   trees
\steven{whose label sets partition}
$X$ is said to be \sthursday{a} {\it forest on $X$}.  Furthermore, we  say that a set  $\mathcal{F}=\{F_1,\dots,F_{k}\}$ of unrooted binary phylogenetic trees -- with $|\mathcal{F}|$ referred to as the {\it size} of $\cF$ -- is a {\it forest for} $T$ if $\mathcal{F}$ can be obtained from $T$ by deleting a $(k-1)$-sized subset $E$ of $E(T)$, suppressing any unlabeled leaves, and then finally suppressing any vertices with degree 2. To ease reading, we write $\cF = T-E$ if $\mathcal{F}$ can be obtained in this way. 

\begin{definition}[uMAF]
\label{def:uMAF}
A set $\mathcal{F}$ of \steven{unrooted}  trees is an {\it agreement forest} for $T_1$ and  $T_2$ \steven{(denoted $uAF$)} if $\mathcal{F}$ is a forest \steven{of both} $T_1$ and $T_2$. 
\steven{An} unrooted \emph{maximum} agreement forest (uMAF), is an uAF of \steven{minimum} size.
\end{definition}

\pagebreak
So, the uMAF problem is defined as follows: \\
\\
\textbf{Problem: $uMAF(T_1, T_2)$}\\
\textbf{Input: }Two unrooted binary trees $T_1$, $T_2$ on the same set of taxa $X$.\\
\textbf{Output: }\steven{An} uMAF for $T_1$ and $T_2$.\\

The two problems defined above are closely related, and known to be NP-hard \cite{AllenSteel2001}.

\begin{theorem}[\cite{AllenSteel2001}]
\label{thm:tbrIsMaf}
Given two unrooted binary trees $T_1$, $T_2$ on the same set of taxa $X$, we have that  $d_{TBR}(T_1, T_2)=|uMAF(T_1, T_2)|-1$.
\end{theorem}

Fortunately, they have been proved to be FPT in their natural parameters \steven{\cite{AllenSteel2001}}, and fast algorithms have been recently proposed \cite{whidden2013fixed,chen2015parameterized}. In this section, we will give a more compact proof of their fixed parameter tractability.


\begin{theorem}
\label{thm:twbound}
Let  $T_1$, $T_2$ be two unrooted binary trees  on the same set of taxa $X$ such that a \smore{uAF} of size $k$ for these two trees exists. Then, the treewidth of their display graph $D$ is at  most $k+1$.  
\label{theo:treewidthuMAF}
\end{theorem}
%
%
\begin{proof}
From \cite{grigoriev2014low}, we know that the display graph of two identical trees has treewidth 2 \steven{(or 1 in the case that both trees consist of a single vertex)}. 
Thus, if we have an \smore{uAF} $\mathcal{F}=\{F_1,\dots,F_{k}\}$ of size $k$, this means that the display graph $D_0$ of  $\mathcal{F}$ \steven{(which we define as the display
graph constructed from two disjoint copies of $\mathcal{F}$)} has $k$ connected components, and treewidth \steven{at most} 2. \steven{This is because the treewidth of a disconnected graph is equal to the largest treewidth ranging over its connected components.} Now, we can construct a tree decomposition of $D$ from the tree decomposition of  $\mathcal{F}$ as follows: suppose $\mathcal{F}$ can be obtained by removing from $T_1$, respectively $T_2$, a subset of edges $K_1$, respectively $K_2$, and suppressing vertices with degree 2 and unlabeled leaves. First, note that we can reintroduce the suppressed vertices (and their corresponding edges) in $\mathcal{F}$, obtaining a new forest $\mathcal{F}'$, without changing the treewidth. Indeed, given an edge $\{u,v\}$ in  $\mathcal{F}$ that corresponded  to a path  $(u,x_1, \cdots, x_j, v)$ before the suppression of the vertices  with degree 2, we know that there exists a bag $B$ in the tree decomposition of $D_0$ such that $u$ and $v$ are in $B$. Then we can add a set of bags $\{B_1, \cdots, B_j\}$ such that $B_1=\{u,x_1,v\}$, $B_2=\{x_1, x_2,v\}$, $\cdots$ , $B_j=\{x_{j-1}, x_j,v\}$, and 
\steven{add edges $\{B, B_1\}$, $\{B_1, B_2\}$, $\cdots$, $\{B_{j-1}, B_j\}$ to the
tree decomposition}.
For the suppressed unlabeled leaves, say $u$, this is even easier: we add a bag  $\{u,v\}$ as child of any of the bags containing $v$, where $v$ is the vertex from which the suppressed leaf was hanging. It is easy to see that this is a  tree decomposition  of the display graph of $\mathcal{F}'$ with treewidth 2. Now, we can easily reintroduce the $k-1$ edges in $K_1$ to the display graph, again without changing the treewidth, by,  for each  edge $\{u,v\}$ in  $K_1$,  adding a bag $\{u,v\}$ between two \smore{existing} bags, one containing $u$ and the other containing $v$. Note that the obtained decomposition is still a tree, since we are connecting two components of $\mathcal{F}'$. Now, when adding back the edges of $K_2$, this is not true anymore. In this case, there exists at least a path in the tree decomposition, connecting a bag containing $u$ to a bag containing $v$. Then, taking the shortest of these paths and adding $u$ to its  bags not containing  $u$, we increase the treewidth \smore{by at most} 1. If we do this for all edges in $K_2$, we obtain a tree decomposition for the display graph of $T_1$ and $T_2$ with treewidth \steven{at most} $2+ (k-1)=k+1$. \steven{Note that this bound is tight, as the following example
shows: an uMAF of two quartets with different topologies, $uv|wx$ and $ux|vw$ say,
contains 2 components, and the display graph of these two quartets has treewidth 3 (see also \cite{grigoriev2014low}).}
$\qed$
\end{proof}

\steven{In the following, we will demonstrate that $|uMAF(T_1,T_2)| =: k$ can be computed in time $O( f(k) \cdot |X| )$ for some computable function $f$ that depends only on $k$. We do this via the machinery of MSOL. The high-level idea is that
we formulate a logical query to answer the question \emph{``Is $k \leq k'$?''} for increasing values of $k'$ until the answer is \emph{yes}, and then
stop: at this point $k' = k$. } \cs{\sthursday{We use the stronger variant} of MSOL \smore{that allows quantification over both edges and vertices}\commentOut{, known as MSO$_2$}. In particular, we will use the \emph{extended} MSOL framework of Arnborg et al \cite{Arnborg91}. Following \cite{compatibility2006,strictcompatibility2014}} we
\steven{note that the sets $V_1, E_1, V_2, E_2, X$ \smore{(and later, $\rho$)} are all available to the MSOL query i.e.
within the query we can distinguish which vertices/edges of $D$ belong to $T_1$, which belong to $T_2$, and which are taxa.}

\steven{More formally,} we construct an MSOL formula \smore{$\Phi(K_1, K_2)$} and a relational structure $\textbf{G}$ such that 
 ${\bf G} \models \smore{\Phi(K_1, K_2)}$ if and only if
\smore{$K_1$ is a set of $k'-1$ edges of $E_1$, and $K_2$ is a set of $k'-1$ edges of $E_2$,}
such that, after deleting them, the resulting components form an \cswg{uAF} $\mathcal{F}$ for both $T_1$ and $T_2$, i.e. $\mathcal{F}_1=T_1-K_1=\mathcal{F}_2=T_2-K_2$. To model this, we need to have that: (1) \steven{the two forests 
$\mathcal{F}_1$ and $\mathcal{F}_2$ induce an identical partition of $X$}
and (2) the components of the two induced forests must have the same topology. To enforce (1) we observe
that (in, say, $T_1$) two taxa
\smore{$x_1$ and $x_2$}
are in the same component of the forest resulting from
deletion of $K_1$ if and only if they can still reach each other inside $T_1$ after
deletion of those edges. \steven{In turn, this occurs} if and only if there is a path from $x_1$ to $x_2$ entirely contained
inside $T_1$ which avoids all the edges in $K_1$. To enforce (2) we demand that
a quartet is in the first forest \steven{(i.e. the quartet is contained inside one of the
trees in the forest)} if and only the quartet is in the second forest.
\steven{This uses the fact that two unrooted binary  trees on the same set of
taxa are topologically identical if and only if they induce identical sets of quartets \cite{buneman1971recovery}}.

Before defining  \smore{$\Phi(K_1, K_2)$}, we need to introduce several intermediate predicates. These
build on a number of basic predicates \steven{which we mainly list for the benefit of readers
not familiar with MSOL}. They are used to: 
\begin{itemize}
\item \steven{test that $Z$ is equal to} the union of two sets $P$ and $Q$: \\
$P \cup Q = Z := \forall z ( z \in Z \Rightarrow z \in P \vee z \in Q)
\wedge
\forall z ( z \in P \Rightarrow z \in Z)
\wedge\forall z ( z \in Q \Rightarrow z \in Z)$\\
\item  test that $P \cap Q = \emptyset$: \\
$NoIntersect(P,Q) := \forall u \in P( u \not \in Q $)\\
\item test that $P \cap Q = \{v\}$:\\
$Intersect(P,Q,v) := (v \in P) \wedge (v \in Q) \wedge \forall u \in P( u \in Q \Rightarrow (u  = v) )$\\
\item test if the sets $P$ and $Q$ are a bipartition of $Z$:  \\
$Bipartition(Z, P, Q) := (P \cup Q = Z) \wedge NoIntersect(P,Q)$\\
\item test if the elements in $\{x_1, x_2, x_3, x_4\}$ are pairwise different: \\
$allDiff( x_1, x_2, x_3, x_4 ) := \bigwedge_{i \neq j \in \{1,2,3,4\}} x_i \neq x_j$\\
\item check if the nodes $p$ and $q$ are adjacent in $D$: \\
$adj(p,q) := \exists e \in E ( R^{D}(e,p) \wedge R^{D}(e,q))$
\end{itemize}


The predicate 
$PAC(Z, x_1, x_2, K_i)$  \steven{(\emph{``path avoids cuts?'}')} asks: is there a path from 
$x_1$ to $x_2$ entirely contained inside \steven{vertices} $Z$ that avoids all the edges $K_i$? We 
model this by observing that this does \emph{not} hold if you can partition $Z$ into two
pieces $P$ and $Q$, with $x_1 \in P$ and $x_2 \in Q$, such that the only edges that
cross the induced cut (if any) are in $K_i$:
\begin{eqnarray*}
PAC(Z, x_1, x_2, K_i) &:=& (x_1 = x_2) \vee \neg \exists P, Q ( Bipartition(Z,P,Q) \wedge x_1 \in P \wedge x_2 \in Q \wedge\\ &&(\forall p, q (p \in P \wedge q \in Q \Rightarrow \neg adj( p,q) \vee (
\exists g \in K_i( R^{D}(g, p)\\
&& \wedge R^{D}(g,q))))))
\end{eqnarray*}

We model
that a quartet is in the forest (of, say, $T_1$) \steven{by stipulating} that there is an embedding \steven{(i.e. subdivision)} of the
quartet, \steven{completely contained inside $T_1$}, which avoids all the edges in $K_1$. To model the embedding, we model
the five edges of the quartet as five subsets of vertices $A,B,C,D,P$, representing the subdivisions of the five edges, with $P$ being the
central edge and $u$ and $v$ being its endpoints. We demand that (with the exception of $u$
and $v$) these subsets are disjoint. \steven{This is all combined in the following $QAC^{1}$ predicate (\emph{``quartet avoids cuts in $T_1$?'')},
which returns true if and only if $T_1$ contains an embedding of $x_a x_b | x_c x_d$ that
is disjoint from the edge cuts $K_1$.}
\begin{eqnarray*}
QAC^{1}(x_a, x_b, x_c,x_d, K_1) &:=&
\exists u, v \in V_1 ( (u \neq v) \wedge \exists A,B,C,D,P \subseteq V_1 ( x_a, u \in A \wedge x_b, u \in B \wedge x_c,\\&& v \in C \wedge x_d,v \in D \wedge u \in P \wedge v \in P \wedge Intersect(A,B,u) \wedge\\ && Intersect(A,P,u) \wedge Intersect(B,P,u) \wedge Intersect(C,D,v) \wedge \\&&  Intersect(C,P,v) \wedge Intersect(D,P,v) \wedge   NoIntersect(A,C) \wedge \\&&  NoIntersect(B,C)\wedge NoIntersect(A,D)\wedge NoIntersect(B,D)   \wedge \\&& PAC(A, u, x_a, K_1)
\wedge PAC(B, u, x_b, K_1)
\wedge PAC(C, v, x_c, K_1) \wedge \\&&
 PAC(D, v, x_d, K_1)  \wedge
 PAC(P, u, v, K_1))) 
\end{eqnarray*}
We can define $QAC^{2}(x_a, x_b, x_c,x_d, K_2)$  in a similar way.  Note that, for every four taxa, we need to consider
all three possible quartet topologies. Then we define \smore{$\Phi(K_1, K_2)$} as follows:
\begin{eqnarray*}
(\bigwedge_{i\in \{1,2\}}|K_i| = k'-1) & \wedge& (\bigwedge_{i\in \{1,2\}}K_i \subseteq E_i)
\wedge \forall x_1, x_2 \in X(PAC(V_1, x_1, x_2, K_1)  \Leftrightarrow \\&& 
 PAC(V_2, x_1, x_2, K_2) ) \wedge \forall x_1, x_2, x_3, x_4 \in X( allDiff(x_1, x_2, x_3, x_4) \Rightarrow \\
&&((QAC^{1}(x_1, x_2, x_3, x_4, K_1) \Leftrightarrow QAC^{2}(x_1, x_2, x_3, x_4, K_2) )\wedge \\
&& (QAC^{1}(x_1, x_3, x_2, x_4, K_1 \Leftrightarrow QAC^{2}(x_1, x_3, x_2, x_4, K_2) ) \wedge \\
&& (QAC^{1}(x_1, x_4, x_2, x_3, K_1) \Leftrightarrow QAC^{2}(x_1, x_4, x_2, x_3, K_2) ))).\\
\end{eqnarray*}
\smore{(The cardinality operator is permitted because the extended MSOL framework of \cite{Arnborg91} allows the incorporation of
an \emph{evaluation relation} which can test, amongst other things, the cardinalities of free set variables).}

\begin{theorem}
Computation of TBR / uMAF on two 
unrooted binary trees on the same set of taxa $X$ is linear time FPT. \steven{That is,
the optimum $k$ can be computed in time $O( f(k) \cdot |X| )$ for some computable
function that only depends on $k$.}
\label{theo:TBR}
\end{theorem}
\begin{proof}
We have presented a logical query to answer the question \emph{``Is $k \leq k'$?''} for increasing values of $k'$. 
For each value of $k'$ the MSOL query, which examines the display graph $D$, has \smore{fixed length}.
Combining this with the fact that the treewidth of $D$ is bounded by
a function of $k$ (by Theorem \ref{thm:twbound}), and that the size of $D$ is a linear function of $|X|$, we have
\smore{the desired result}. (Note that the actual edge cuts - which can be used
to construct a uMAF - can also be obtained in the same time bound by leveraging Theorem 4
of \cite{compatibility2006}.)
\qed

\end{proof}

\commentOut{
Note that one could give an alternative logical formulation without free variables, if desired. That is, one could
\sblah{completely remove $K_1$ and $K_2$ and the cardinality test from the formulation by instead
explicitly introducing nested existential quantifiers, selecting individual edges, to a depth of $2(k'-1)$. The usual
logical constructions can then be used to test that the edges are distinct, that $k'-1$ of them are from $E_1$,
and $k'-1$ are from $E_2$.} Such a formulation
has the disadvantage that the length of the logical query grows as a function of $k'$. However, the length of the query remains bounded by a function of $k$, so
the same FPT result would be obtained in this case.}

\subsection{rSPR / MAF on rooted trees}

In this section, we will give a compact proof that the computation of rSPR distance is FPT in \steven{its} natural parameter. Before that, we need to introduce some definitions.

\begin{definition}[rSPR move]
Given a rooted binary tree $T$,  a \emph{subtree prune and regraft} (rSPR) move on $T$ consists of removing an edge of $T$, say $(u,v)$, yielding two trees $T_u$ and $T_v$,
and then reconnecting them as follows: subdividing some edge of \smore{$T_u$} with a new vertex p; adding an edge
directed from $p$ to \smore{$v$}, and then suppressing any vertices with indegree and outdegree both equal to 1.
\end{definition}

\steven{rSPR distance is defined analogously to TBR distance,} and a MAF for two rooted binary trees $T_1$, $T_2$ is defined \steven{similarly to} a uMAF, but in a rooted framework.  \steven{We refer to \cite{Bordewich2004} for precise definitions. The
main difference is that a forest consists of \emph{rooted} binary trees and this has to
be taken into account when comparing the topology of the components.} 
\steven{In the rooted context MAFs are mainly studied because of their close relationship
to rSPR distance. To accurately model rSPR distance it is necessary to slightly modify each
input tree $T_i$ as follows: we add a vertex with special label $\rho$ at the end of a pendant edge adjoined to the original root of $T_i$, \sthursday{see Figure \ref{fig:displaygraph}(c).} We then consider $\rho$ to be part of the label
set of the tree. Note that the addition of $\rho$ means that we can equivalently view each $T_i$ as an unrooted binary tree, with $\rho$ acting as a placeholder for the root location, and
this is how the trees will be modelled in the display graph.}

\steven{The close relationship between MAF (assuming $\rho$
has been added as described) and rSPR distance  is summarized by the following well-known result.}

\begin{theorem}[\cite{Bordewich2004}]
Given two rooted binary trees $T_1$, $T_2$ on the same set of taxa $X$, we have that  $d_{rSPR}(T_1, T_2)=|MAF(T_1, T_2)|-1$.
\end{theorem}

Note that these problems have been proved NP-hard  and FPT in their natural parameter \cite{Bordewich2004}.

The \steven{MSOL} formulation is similar to the TBR formulation, but with the following changes.  When checking that the components \steven{induced by the edge cuts partition the taxa in the same way in both $T_1$ and $T_2$} (i.e. by considering pairs of taxa that still have a path between them), we need to range over $X \cup \steven{\{ \rho \}}$ instead of just $X$. More significantly, we need \steven{predicates for} triplets instead of quartets, because we are working in the rooted environment \steven{and two rooted binary trees are topologically equivalent if and only if they contain the same set of triplets \cite{bryant1995extension}. Fortunately we can use the fact that
triplet $xy|z$ is in $T_i$ ($x,y,z \in X$) if and only if quartet $xy|{\rho}z$ is in the unrooted interpretation of $T_i$.}

 However we cannot simply use $\rho$ as the fourth parameter $x_d$ to
\steven{$QAC^{i}$} because this will evaluate to false if the path from $\rho$ to the rest of the quartet \steven{embedding} has been cut. This is not what we need: $\rho$ is \steven{in this context} only there to indicate direction,
so its particular arm of the quartet embedding can be cut without consequence. We can remedy this by introducing \steven{predicates} $Quartet^i$ and $Triplet^i$ which check whether the corresponding
quartet/triplet was in the \emph{original} tree (i.e. before the edge cuts). \steven{We can
then leverage the fact that, if three distinct taxa $x,y,z$ are in the same component of the
forest, the unique triplet topology they induce within the component will be the same topology as they induced
in the original tree.}\\
\\
\steven{We first need the following predicate, which is} a specialization of \steven{the earlier $PAC$ predicate}. It tests whether there is a path from $x_1$ to $x_2$ that is entirely contained inside vertex set $Z$: 
\begin{eqnarray*}
path(Z, x_1, x_2) &:=& (x_1 = x_2) \vee \neg \exists P, Q ( Bipartition(Z,P,Q) \wedge x_1 \in P \wedge x_2 \in Q\\&&
 \wedge (\forall p, q (p \in P \wedge q \in Q \Rightarrow \neg adj( p,q) )))
\end{eqnarray*}

For each tree $T_i$, the following \steven{predicate} checks whether the quartet $x_a x_b | x_c x_d$ is \steven{contained in} $T_i$:
\begin{eqnarray*}
Quartet^{i}(x_a, x_b, x_c,x_d) &:=& 
\exists u, v \in V_i ( (u \neq v) \wedge \exists A,B,C,D,P \subseteq V_i ( x_a, u \in A \wedge \\
&& x_b,  u \in B \wedge x_c, v \in C \wedge x_d,v \in D \wedge u \in P \wedge v \in P \wedge Intersect(A,B,u) \wedge \\
&& Intersect(A,P,u) \wedge Intersect(B,P,u) \wedge Intersect(C,D,v)  \wedge \\
&& Intersect(C,P,v) \wedge Intersect(D,P,v) \wedge NoIntersect(A,C) \wedge  \\
&& NoIntersect(B,C)\wedge NoIntersect(A,D)\wedge NoIntersect(B,D)  \wedge \\
&& path(A, u, x_a) 
\wedge path(B, u, x_b)
\wedge path(C, v, x_c)
\wedge path(D, v, x_d) \wedge \\
&&  path(P, u, v)))
\end{eqnarray*}

For each rooted tree $T_i$, the following \steven{predicate} checks whether the triplet $x_a x_b | x_c$ is \smore{contained in} $T_i$ (simply by checking whether $x_a x_b | x_c \rho$ is \steven{contained in} it): 
\begin{eqnarray*}
Triplet^{i}(x_a, x_b, x_c) := Quartet^{i}( x_a, x_b, x_c, \rho)\\
\end{eqnarray*}
Now, we are ready to define $TAC^{i}$ \steven{(\emph{``triplet avoids cuts in $T_i$?''})}, which models whether a triplet is in the forest of $T_i$ \steven{induced by the edge cuts}:

\begin{eqnarray*}
TAC^{i}(x_a, x_b, x_c,  K_1) &:=& Triplet^{i}(x_a, x_b, x_c) \wedge 
\exists u, v \in V_i ( (u \neq v) \wedge \exists A,B,C,D,P \subseteq V_i ( x_a, u \in A \wedge \\
&&  x_b, u \in B \wedge x_c, v \in C \wedge \rho,v \in D \wedge u \in P \wedge v \in P \wedge Intersect(A,B,u) \wedge \\
&&  Intersect(A,P,u) \wedge Intersect(B,P,u) \wedge Intersect(C,D,v) \wedge \\
&&  Intersect(C,P,v) \wedge Intersect(D,P,v) \wedge NoIntersect(A,C) \wedge \\
&& NoIntersect(B,C)\wedge NoIntersect(A,D)\wedge NoIntersect(B,D) \\ 
&& \wedge PAC(A, u, x_a,  K_1)
\wedge PAC(B, u, x_b, K_1)
\wedge PAC(C, v, x_c,  K_1) \wedge\\
&& 
 path(D, v, \rho,  K_1) 
\wedge PAC(P, u, v,  K_1))) 
\end{eqnarray*}

\smore{Note how we use \emph{path} rather than $PAC$ to model the path from $v$ to $\rho$ i.e. because it does not matter for the triplet whether this path is cut.}
\smore{The final MSOL formulation is then very similar to that given in Section \ref{subsec:tbr}:}


\begin{eqnarray*}
(\bigwedge_{i\in \{1,2\}}|K_i| = k'-1) &\wedge& (\bigwedge_{i\in \{1,2\}}K_i \subseteq E_i)
\wedge \forall x_1, x_2 \in X \cup \{ \rho \}(PAC(V_1, x_1, x_2, K_1) \Leftrightarrow \\&& 
PAC(V_2, x_1, x_2, K_2) ) \wedge  \forall x_1, x_2, x_3 \in X( allDiff(x_1, x_2, x_3) \Rightarrow \\
&&((TAC^{1}(x_1, x_2, x_3, K_1) \Leftrightarrow TAC^{2}(x_1, x_2, x_3, K_2) )\wedge \\
&& (TAC^{1}(x_1, x_3, x_2,  K_1) \Leftrightarrow TAC^{2}(x_1, x_3, x_2,  K_2) ) \wedge \\
&& (TAC^{1}(x_2, x_3, x_1,  K_1) \Leftrightarrow TAC^{2}(x_2, x_3, x_1, K_2) ))).\\
\end{eqnarray*}

\begin{theorem}
Computation of rSPR / MAF on two 
rooted binary trees on the same set of taxa $X$ is linear time FPT. 
\steven{That is,
the optimum $k$ can be computed in time $O( f(k) \cdot |X| )$, for some computable
function that only depends on $k$.}
\end{theorem}
\begin{proof}
\steven{An agreement forest of the two rooted trees $T_1$ and $T_2$ induces
an agreement forest (consisting of unrooted binary trees) of the same size of the unrooted interpretations of these
trees, simply by ignoring the orientation of edges. Hence the treewidth bound
described in Theorem  \ref{theo:treewidthuMAF} is still applicable, and} the theorem follows. (Again, if required one can obtain the actual edge cuts, which can be used to build a MAF, in the same time bound by leveraging Theorem 4
of \cite{compatibility2006}). \qed
\end{proof}

\subsection{Hybridization Number}

In this section, we deal again with rooted trees, and thus we add a vertex \sthursday{labeled $\rho$ to both trees to} indicate the root location, as done \sthursday{for rSPR; see Figure \ref{fig:displaygraph}(c)}. A {\it rooted  phylogenetic network} (\cs{rooted network for short}) $N=(V(N),E(N))$ on a set of taxa $X$ is any rooted acyclic digraph in which no vertex has degree 2 (except possibly the root) and whose leaves are bijectively labeled by elements of $X$. The {\it hybridization number} of $N$, denoted by $h(N)$,  is defined as
\[h(N) = \sum_{\substack{v\in V(N): \delta^-(v)>0}}(\delta^-(v)-1) = |E(N)| - |V(N)| + 1\]
\steven{where $\delta^{-}(v)$ denotes the indegree of $v$.}

Given a rooted  network $N$ on $X$ and  a rooted binary  tree 
$T$ on  $X'$, with $X'\subseteq X$, we say that $T$ is {\it displayed} by $N$ if $T$ can be obtained from $N$ by
deleting a subset of its edges and any resulting degree 0 vertices, and then 
  suppressing vertices with $\delta^-(v)=\delta^+(v)=1$.
 
We are now ready to define the hybridization \steven{number} problem: \\
\\
\textbf{Problem: $HN(T_1, T_2)$}\\
\textbf{Input: }Two rooted binary trees $T_1$, $T_2$ on the same set of taxa $X$.\\
\textbf{Output: }A rooted network \smore{$N$} displaying $T_1$ and $T_2$ such that $h(N)$ is \steven{minimum} over all rooted  networks with this property.\\

The hybridization number for  $T_1$ and $T_2$, denoted by $h(T_1, T_2)$, is defined as the hybridization number of this \steven{minimum} network. As done for TBR and rSPR, we can give a characterization of the hybridization number in terms of agreement forests. To do so, we need to define \emph{acyclic} \steven{agreement forests}.

Let
\smore{$\cF=\{F_1,F_2,\ldots,F_k\}$}
be an agreement forest for two rooted binary  trees $T_1$ and $T_2$ on the same set of taxa $X$, and let
$AG(T_1,T_2,\cF)$ be the directed graph whose vertex set is $\cF$ and for which $(F_i,F_j)$ is an arc iff $i\ne j$,
and either
\begin{enumerate}
\item [(1)] the root of $T_1(\mathcal{L}(F_i))$ is an ancestor of the root of  $T_1(\mathcal{L}(F_j))$  in $T_1$, or
\item [(2)] the root of $T_2(\mathcal{L}(F_i))$ is an ancestor of the root of  $T_2(\mathcal{L}(F_j))$  in $T_2$.
\end{enumerate}
We call  $\cF$ an \emph{acyclic agreement forest} (AAF) for $T_1$ and $T_2$ if $AG(T_1,T_2,\cF)$ does not contain any directed cycle. A \emph{maximum} acyclic agreement forest (MAAF), is an AAF of {minimum} size.

The acyclicity condition is used to model the fact that species cannot inherit genetic material from their own \smore{offspring}. 
The two problems defined above are closely related, \steven{as the following well-known
result shows.}

\begin{theorem}[\cite{BaroniEtAl2005}]
Given two rooted binary trees $T_1$, $T_2$ on the same set of taxa $X$, we have that  \steven{$h(T_1, T_2)=|MAAF(T_1, T_2)| - 1.$}
\end{theorem}

The above equivalence formed the basis for results proving that both problems are NP-hard~\cite{Bordewich2007} and fixed parameter tractable~\cite{sempbordfpt2007}.

Here we show \steven{an alternative proof} that computation of hybridization number on two rooted binary trees with the same set of taxa $X$ is FPT, again using MSOL.
We will do this by demonstrating that $|MAAF(T_1,T_2)| =: k$ can be computed in time $O( f(k) \cdot |X| )$ for some computable function $f$ that depends only on $k$. 
Again, we will formulate a logical query on the  display graph to answer the question \emph{``Is $k \leq k'$?''} for increasing values of $k'$, until $k'=k$ is reached and the answer
to the query is ``yes''. Unlike the formulations given earlier for TBR and rSPR, the query has no free variables, and the length of each query will grow as a function of $k'$. However, given that $k' \leq k$, the
length will remain bounded by a function of $k$.
Note that, if a MAAF of size $k$ exists for $T_1$ and $T_2$, then an AF of size $k$ exists too, \steven{and as argued \sthursday{for rSPR}, if two rooted trees have an agreement
forest of size $k$ then so do the underlying, unrooted trees.} So the \smore{treewidth} bound of Theorem \ref{theo:treewidthuMAF} is still valid, where $k=|MAAF(T_1,T_2)|$, and this implies that  an overall running time of the form $O( f(k) \cdot |X| )$ can be achieved. 

\steven{The major challenge when modelling MAAF is to encode the acyclicity constraints. It is not clear whether the formulations from the previous sections, in which agreement forests are modelled directly as sets of edge-cuts, can be (elegantly) extended to include acyclicity constraints. For this reason we choose to discard the agreement forest abstraction,
using it only to generate the treewidth upper bound. For the actual modelling we
use an alternative ``elimination-ordering'' characterization of MAAF/HN, first presented in \cite{kelk2012elusiveness}, which we briefly summarize here.}

\steven{Given a rooted binary tree $T$ on $X$, we say a subtree $T'$ of $T$ is \emph{pendant}
if there exists a vertex $u$ of $T$ such that $T' = T_u$. In this case it is then natural to
associate $T'$ with the subset of $X$ labeling its leaves, i.e. $\mathcal{L}(T')$. We say that $T'$ is a \emph{common pendant} subtree of $T_1$ and $T_2$ if it is a pendant subtree of both $T_1$ and $T_2$. We call $(S_{1}, S_{2}, \dots, S_{p})$ $(p \geq 0)$ a \emph{common pendant subtree sequence} of $T_1$ and $T_2$ of length $p$ if for every $1 \leq i \leq p$, $S_i$ is a common pendant subtree of $T_1 - \cup_{j<i} \mathcal{L}(S_j)$ and
$T_2 - \cup_{j<i} \mathcal{L}(S_j)$. We say that such a sequence is additionally a \emph{tree sequence} if the two trees $T_1 - \cup_{j \leq p} \mathcal{L}(S_j)$ and
$T_2 - \cup_{j \leq p} \mathcal{L}(S_j)$ are identical. Informally, a tree sequence of length
$p$ describes a sequence of $p$ common pendant subtrees that can be successively pruned from the original trees to reach a common core tree. If $T_1$ and $T_2$ are already identical then
we use the empty tree sequence $\emptyset$, and take $p=0$, to represent this.}

\steven{
The results in \cite{kelk2012elusiveness} establish that $h(T_1, T_2)$ is equal to the smallest $p$ such that a tree sequence of length $p$ exists. This is the characterization of optimality that we will
use i.e. each logical query will pose the question, \emph{``Does a tree sequence of length $k'$ exist?''}. There is no need to model acyclicity in this formulation. However, we do need to model the concept \emph{common pendant subtree} and the impact of earlier pruning steps on the original trees. 
}

 
Before writing down the MSOL formulation we need some new \steven{auxiliary predicates}.  
The first \steven{predicate} checks whether there is a path from $x_1$ to $x_2$ within $Z$ that survives the deletion of vertex $u$. \sthursday{This is similar to the $PAC$ predicate defined earlier.}
\commentOut{
Namely, some path survives the
deletion if and only if \smore{$x_1$ and $x_1$ both survive and} every bipartition with $x_1$ and $x_2$ on different sides of the bipartition, has some edge crossing the bipartition that does not have $u$ as an endpoint. (Equivalently: the
path does not survive if and only if you can find a bipartition in which all crossing edges have been killed by deletion of $u$).}
\begin{eqnarray*}
pathSurvivesVertexCut(Z, x_1, x_2, u) &:=& \smore{(u \neq x_1) \wedge (u \neq x_2) \wedge}\\&&((x_1 = x_2) \vee \neg \exists P, Q ( Bipartition(Z,P,Q) \wedge \\&&x_1 \in P \wedge x_2 \in Q \wedge (\forall p, q (p \in P \wedge q \in Q \Rightarrow \\
&&\neg adj( p,q) \vee p=u \vee q=u))))\\
\end{eqnarray*}

\steven{For a vertex $u \neq \rho$ in a tree $T_i$ and a taxon $x \in X$, observe that
$x$ is in the clade rooted at $u$ (i.e. in the label set of the pendant subtree rooted at $u$) if and only if $(x=u)$ or deleting $u$ from $T_i$ destroys all paths from $\rho$ to $x$ (inside $T_i$). Hence:}
\[
InCladeUnder^{i}( u, x ) := (u = x) \vee \neg pathSurvivesVertexCut(V_i, \rho, x, u)
\]
This leads naturally to a \smore{predicate} for testing whether $C \subseteq X$ is a clade of $T_i$:
\[
Clade^{i}(C) := \exists u \in V_i (\forall x ( x \in C \Leftrightarrow InCladeUnder^{i}(u,x)))
\]
As we shall see, it is useful to extend this \steven{predicate} with an optional list $Z_1, Z_2, \ldots$ which represent subsets of $X$ \steven{describing common pendant subtrees that have already been pruned from the tree}. The
statement $Clade^{i}(C, Z_1, Z_2, \ldots...)$ evaluates to true if \smore{and only if} $C$ is a clade of $T_i$
\emph{after} the taxa in $Z_1, Z_2, \ldots...$ have been pruned away. \steven{(To avoid ambiguity the predicate automatically returns false if $C$ intersects with any of the $Z_i$.)} Note that \steven{the list of $Z_i$ is shown in square brackets to emphasize that it is a ``macro'':
there will be a different predicate for each possible list length $t$. The list of $Z_i$ will never be longer than $h(T_1, T_2)$, and length of the generated predicate will
be bounded by a function of the list length, so the length of the overall logical query remains
bounded by a function of $h(T_1, T_2)$.}
\begin{eqnarray*}
Clade^{i}(C, [Z_1, \ldots..., Z_t]) &:=& NoIntersect(C, Z_1) \wedge ... \wedge NoIntersect(C, Z_t) \wedge \exists u \in V_i \\
&& (\forall x ( x \in C \Rightarrow InCladeUnder^{i}(u,x)) \wedge \\
&& \forall x (InCladeUnder^{i}(u,x) \Rightarrow (x \in C) \vee (x \in Z_1) \vee ... \vee (x \in Z_t)))
\end{eqnarray*}

\steven{We are now ready to define the CPS (i.e. ``common pendant subtree'') predicate. We
do this by observing that $C \subseteq X$ corresponds to a common pendant subtree of $T_1$ and $T_2$ if
and only if $C$ is a clade of both trees (this ensures that $C$ is pendant in both trees)
and the set of triplets induced by $C$ is identical in both trees (this ensures that the
pendant subtree has the same topology in both trees).}
\begin{eqnarray*}
CPS(T_1, T_2, C) &:=& Clade^{1}(C) \wedge Clade^{2}(C)
\wedge \forall x \forall y \forall z (
x,y,z \in C \wedge \\ &&\ allDiff(x,y,z) \Rightarrow (Triplet^{1}(x,y,z) \Leftrightarrow Triplet^{2}(x,y,z)))
\end{eqnarray*}

We extend this now with \steven{a list of $Z_i$} representing the taxa we have already pruned. \steven{This new version of the predicate} evaluates to true \steven{if and only if} $C$ \steven{corresponds to} a common pendant subtree in the two trees \emph{after} all the $Z_i$ have been pruned away. (Here we make implicit use of the fact that \steven{the}
\emph{Clade} \steven{predicate} immediately returns false whenever $C$ intersects with the $Z_i$.)
\begin{eqnarray*}
CPS(T_1, T_2, C, [Z_1, \ldots, Z_t]) &:=& 
Clade^{1}(C, Z_1, \ldots, Z_t) \wedge   Clade^{2}(C, Z_1, \ldots, Z_t)
\wedge  \forall x \forall y \forall z \\
&& (
x,y,z \in C \wedge allDiff(x,y,z) \Rightarrow (Triplet^{1}(x,y,z) \Leftrightarrow  \\ 
&&Triplet^{2}(x,y,z))
)
\end{eqnarray*}



\steven{We are now ready to directly pose the question: is there a tree sequence of length
$k'$? We can assume $k' \geq 1$ because $k' =0$ is trivial to check in polynomial time. To
make the formulation slightly more compact we actually construct a list of length $k' + 1$,
where $C_{k'+1}$ represents the taxa that still remain after the common pendant subtrees have been pruned away: we
can then test that the sequence is a \emph{tree} sequence (i.e. that a common core tree remains) by testing
that $CPS(T_1, T_2, C_{k'+1}, C_1, \ldots, C_{k'})$ is true. Note that the \emph{HybNum}
predicate is again a macro, whose expansion depends on $k'$.} 
\begin{eqnarray*}
HybNum[k'](T_1, T_2) &:=& \exists C_1, \ldots, C_{k'}, \steven{C_{k'+1}} ( Partition(X, C_1, \ldots, C_{k'+1}) \wedge\\
&& CPS(T_1, T_2, C_1, \emptyset) \wedge\\
&& CPS(T_1, T_2, C_2, C_1) \wedge\\
&&CPS(T_1, T_2, C_3, C_1, C_2) \wedge\\
&&...\\
&& CPS(T_1, T_2, C_{k'+1}, C_1, \ldots, C_k')).
\end{eqnarray*}

The \emph{Partition} predicate has the expected meaning and definition:\\
\\
\sthursday{$Partition(X, C_1, \ldots, C_k) := (\wedge_{i \neq j} NoIntersect(C_i, C_j)) \wedge \forall u ( u \in X \Leftrightarrow (u \in C_1 \vee u \in C_2 \vee \ldots \vee  u \in C_k))$}\\
\\
\cs{Concluding, we have the following result:}
\begin{theorem}
\steven{Computation of hybridization number / MAAF on two 
rooted binary trees on the same set of taxa $X$ is linear time FPT. 
That is, the optimum $k$ can be computed in time $O( f(k) \cdot |X| )$, for some computable
function that only depends on $k$.}
\end{theorem}

\subsection{Parsimony distance on binary characters}

Let $T$ be an unrooted binary  tree on a set of taxa $X$. A \emph{binary character} $f$ is simply a function $f : X \rightarrow \{red,blue\}$. An \emph{extension} of $f$ to $T$ is
a mapping $g : V(T) \rightarrow \{red,blue\}$ such that, for all $x \in X$, $g(x)=f(x)$. For
a given character $f$, an \emph{optimal extension} is any extension $g$ of $f$ such
that the number of bichromatic edges is minimized. The number of bichromatic edges
in an optimal extension is called the \emph{parsimony score} of $f$ with respect to $T$, and denoted $l_f(T)$. The well-known algorithm by Fitch can be used to compute $l_f(T)$ (and an optimal extension) in polynomial time \cite{fitch1971}. We shall describe Fitch's algorithm in due course. The \emph{parsimony distance problem on binary characters}, denoted $d^2_{MP}$, is
defined as follows \cite{fischer2014maximum}.\\
\\
\textbf{Problem: $d^2_{MP}(T_1, T_2)$}\\
\textbf{Input: }Two unrooted binary trees $T_1$, $T_2$ on the same set of taxa $X$\\
\textbf{Output: }Construct a binary character $f$ on $X$ such that the
value $| l_f(T_1) - l_f(T_2) |$ is maximized.\\
\\
We use $d^2_{MP}$ to denote
the optimum value of $| l_f(T_1) - l_f(T_2)|$. The problem was recently shown to be NP-hard \steven{and APX-hard} \cite{kelk2014complexity}. It is not known whether the problem is FPT in $d^2_{MP}$. The following result, however, is already known.

\begin{lemma}[\cite{fischer2014maximum}]
\label{lem:fk2014}
Let $T_1, T_2$ be two unrooted binary trees on the same set of taxa $X$. Then
$d^2_{MP}(T_1, T_2) \leq d_{TBR}(T_1, T_2)$.
\end{lemma}

Given two trees $T_1, T_2$ as input to $d^2_{MP}$, it is not known whether the display
graph $D$ \cs{of $T_1$ and  $T_2$} has treewidth bounded by a function of $d^2_{MP}$. However, 
from Lemma \ref{lem:fk2014} and earlier results in this article \steven{(Theorems \ref{thm:tbrIsMaf} and \ref{thm:twbound})} it is clear that $D$
has treewidth bounded by a function of $d_{TBR}(T_1, T_2)$. An MSOL formulation modelling
$d^2_{MP}$, whose length is bounded by a function of $d^2_{MP}$, will therefore give a running time of the form $f( d_{TBR}(T_1, T_2) ) \cdot O(|X|)$ 
for some computable function \steven{$f$} that only depends on $d_{TBR}(T_1, T_2)$. We now give such a formulation. We will remain within the framework of \cite{Arnborg91}, this time using the (``linear extremum'') optimization variant of MSOL. \steven{This allows us to maximize or minimize an affine
function of (the cardinalities of) the free set variables in the query.}

The MSOL formulation we give here, which is based on an ILP formulation from
\cite{kelk2014complexity}, maximizes $l_f(T_1) - l_f(T_2)$. (To compute $d^{2}_{MP}$ we
need to use the MSOL machinery twice, once for $l_f(T_1) - l_f(T_2)$ and once
for $l_f(T_2) - l_f(T_1)$, taking the maximum of the two results. The second call only
differs in its objective function so we omit details).


The basic idea is to range over all possible binary characters, simultaneously embedding two static formulations\footnote{Interestingly, the earlier phylogenetics MSOL
articles \cite{compatibility2006,strictcompatibility2014} also used static formulations: in that
case the classical polynomial-time algorithm of Aho.} of Fitch's algorithm to ``compute'' $l_f(T_1) - l_f(T_2)$.  

Fitch's algorithm proceeds as follows. If $T$ is not rooted, we root it arbitrarily (by subdividing an arbitrary edge). The algorithm then works in two phases, a bottom-up phase which computes
$l_f(T)$, and then a top-down phase which actually computes a corrresponding extension. In
the bottom-up phase, we start by assigning each taxon $x$ the singleton set of colours $S(x) := \{ f(x) \}$. For an internal node $u$ with children $v_1, v_2$ we set $S(u) := S(v_1) \cap S(v_2)$ (if $S(v_1) \cap S(v_2) \neq \emptyset$, in which case we say $u$ is an \emph{intersection node}) and $S(u) := S(v_1) \cup S(v_2)$ (if $S(v_1) \cap S(v_2) = \emptyset$, in which case we say that $u$ is a \emph{union node}). The value $l_f(T)$ is equal to the number of internal nodes that are union nodes. (We omit a description of the constructive top-down phase as it is not relevant for this article).

To translate this into an MSOL formulation, we begin by arbitrarily rooting $T_1$ and $T_2$ and \steven{using $\rho$ as the placeholder for the root, in the usual fashion.} The central idea
is to partition the vertices of each tree $T_i$ into four possible \steven{subsets $R^{i}, B^{i}, RB^{i}_{I}$ and $RB^{i}_{U}$} corresponding to the set \steven{of colours that Fitch allocates to each node, and distinguishing union events from intersection events:}  \emph{red}, \emph{blue}, \emph{\{red, blue\}} (intersection node) and \emph{\{red, blue\}} (union node). 
\steven{We therefore ask the MSOL
formulation to instantiate the free set variables $R^{i}, B^{i}, RB^{i}_{I}$ and $RB^{i}_{U}$ $(i \in \{1,2\})$ such that the expression $|RB^{1}_U| - |RB^{2}_U|$ is maximized. }
(If desired, this can then be made constructive via Theorem 4
of \cite{compatibility2006}.) The only significant work is \steven{simulating the bottom-up execution of Fitch's algorithm}. In particular, encoding expressions which describe the state of a parent node $u$ in terms of its two children $v_1, v_2$.\\
\\
We introduce the \steven{auxiliary} predicate $child^{i}(u,v)$ which says that $v$ is a child of $u$
in $T_i$. We can model this as follows: $v$ is a child of $u$ in $T_i$ if and only if there is an edge $e$ in $T_i$ such that $v$ and $u$ are both endpoints of $e$ and there does \emph{not} exist a path
from $\rho$ to $v$ that survives the edge cut $e$. \steven{(Here we have specialized the PAC predicate from earlier so that it only takes
a single edge, rather than a set of edges, as its \smore{fourth} argument.)}
\[
child^{i}(u,v) := (u \neq v) \wedge \exists e \in E_i ( R^{D}(e,u) \wedge R^{D}(e,v) \wedge \neg PAC(\smore{V_i}, \rho, v,  e )) 
\]
\steven{For each tree $T_i$ we add the following \smore{constraints}, which encode (in this
order):}
\steven{
\begin{itemize}
\item The four subsets $R$, $B$, $RB_{I}$ and $RB_{U}$ partition the vertices of the tree;
\item A vertex in $X$ can only be in $R$ or $B$;
\item An internal node is in $R$ if and only if (one child is in $R$ and the other child is not
in $B$);
\item An internal node is in $B$ if and only if (one child is in $B$ and the other child is not
in $R$);
\item An internal node is in $RB_I$ if and only if (neither child is in $R$ or $B$);
\item An internal node is in $RB_U$ if and only if (one child is in $R$ and one child is in $B$).
\end{itemize}
}
\begin{eqnarray*}
\bigg ( Partition( V_i, R^{i}, B^{i}, RB^{i}_{I}, RB^{i}_{U} ) \bigg ) \wedge
\bigg ( \forall x \in X( x \not \in RB^{i}_{I}  \wedge x \not \in RB^{i}_{U} ) \bigg ) \wedge\\
\bigg ( \forall u \in V_i( u \not \in X \Rightarrow ( u \in R^{i} \Leftrightarrow \exists c_1, c_2 \in V_i( (c_1 \neq c_2) \wedge child^{i}(u, c_1) \wedge child^{i}(u,c_2) \wedge \\ c_1 \in R^{i} \wedge c_2 \not \in B^{i}    ))) \bigg ) \wedge \\
\bigg ( \forall u \in V_i( u \not \in X \Rightarrow ( u \in B^{i} \Leftrightarrow \exists c_1, c_2 \in V_i( (c_1 \neq c_2) \wedge child^{i}(u, c_1) \wedge child^{i}(u,c_2) \wedge \\ c_1 \in B^{i} \wedge c_2 \not \in R^{i}    ))) \bigg ) \wedge \\
\bigg ( \forall u \in V_i( u \not \in X \Rightarrow ( u \in RB^{i}_{I} \Leftrightarrow \exists c_1, c_2 \in V_i( (c_1 \neq c_2) \wedge child^{i}(u, c_1) \wedge child^{i}(u,c_2) \wedge \\ c_1 \not \in R^{i} \wedge c_1 \not \in B^{i} \wedge c_2 \not \in R^{i} \wedge c_2 \not \in B^{i}  ))) \bigg ) \wedge \\
\bigg ( \forall u \in V_i( u \not \in X \Rightarrow ( u \in RB^{i}_{U} \Leftrightarrow \exists c_1, c_2 \in V_i( (c_1 \neq c_2) \wedge child^{i}(u, c_1) \wedge child^{i}(u,c_2) \wedge \\ c_1 \in R^{i} \wedge c_2 \in B^{i} ))) \bigg ).
\end{eqnarray*}
\steven{Finally, we ensure that both trees select the same character as follows:}
\[
\sthursday{
\forall x \in X( (x \in R^{1} \Leftrightarrow x \in R^{2}) \wedge (x \in B^{1} \Leftrightarrow x \in B^{2}))}
\]
\steven{This concludes the formulation}.
\cs{Then we have the following result:}

\begin{theorem}
$d^2_{MP}(T_1, T_2)$ is \smore{linear time} fixed parameter tractable in parameter $d_{TBR}(T_1,T_2)$.
\end{theorem}

\section{Conclusion}

\steven{We have demonstrated how agreement forests, which are intensively studied objects in the phylogenetics literature, naturally lead to bounded treewidth in an auxiliary graph structure known as the \emph{display graph}. This opens the door to compact, ``declarative'' proofs of fixed parameter tractability for a range of phylogenetics problems by formulating them in Monadic Second Order Logic (MSOL). Our formulations have introduced a number of logical predicates and design principles that will hopefully be of use to other phylogenetics researchers seeking to utilize this powerful machinery elsewhere in phylogenetics. Indeed, a natural follow-up question is to ask: what are the essential
characteristics of phylogenetics problems that are amenable to this technique?}
\section{Acknowledgements}

We thank Mathias Weller for helpful conversations.

\bibliographystyle{plain}

\bibliography{courcelleArxiv}

\end{document}